\documentclass[letterpaper, 11pt]{article}
\usepackage{appendix}
 
\newcommand{\comment}[1]{}

\usepackage{graphicx}
\usepackage{subfigure}
\usepackage{amsmath}
\usepackage{mdwlist}
\usepackage{pxfonts}
 
\usepackage{algorithm}

\usepackage{hyperref}

\usepackage{xspace}

\newenvironment{proof}{\paragraph{\bf Proof:}}{\hspace*{\fill}\(\Box\)}

\newtheorem{theorem}{Theorem}

\newtheorem{claim}{Claim}

\newtheorem{lemma}{Lemma}

\def\noflash#1{\setbox0=\hbox{#1}\hbox to 1\wd0{\hfill}}

\newcommand{\scriptf}{\mathcal{F}}

\newcommand{\scriptv}{\mathcal{V}}

\begin{document}
\title{Byzantine Broadcast\\ Under a {\em Selective Broadcast} Model\\ for Single-hop Wireless Networks\footnote{\normalsize This research is supported in part by National Science Foundation award CNS 1059540. Any opinions, findings, and conclusions or recommendations expressed here are those of the authors and do not necessarily reflect the views of the funding agencies or the U.S. government.}}

\author{Lewis Tseng \hspace*{1in} Nitin Vaidya\\
\normalsize
 University of Illinois at Urbana-Champaign\\ \normalsize Email: \{ltseng3, nhv\}@illinois.edu}

\date{May 2012}
\maketitle

\begin{abstract}
\normalsize

This paper explores an old problem, {\em Byzantine fault-tolerant Broadcast} (BB), under a new model, {\em selective broadcast model}. The new model ``interpolates'' between the two traditional models in the literature. In particular, it allows fault-free nodes to exploit the benefits of a broadcast channel (a feature from reliable broadcast model) and allows  faulty nodes to send mismatching messages to different neighbors (a feature from point-to-point model) simultaneously. The {\em selective broadcast} model is motivated by the potential for {\em directional} transmissions on a wireless channel.

 We provide a collection of results for a single-hop wireless network under the new model. First, we present an algorithm for {\em Multi-Valued} BB that is order-optimal in bit complexity. Then, we provide an algorithm that is designed to achieve BB efficiently in terms of message complexity. Third, we determine some lower bounds on both bit and message complexities of BB problems in the {\em selective broadcast model}. Finally, we present a conjecture on an ``exact'' lower bound on the bit complexity of BB under the {\em selective broadcast} model.

\end{abstract}

\thispagestyle{empty}
 
 
\setcounter{page}{1}

\section{Introduction}
\label{sec:intro}

In this paper, we address Byzantine fault-tolerant broadcast (or {\em Byzantine Broadcast} for short) under a new model of a wireless broadcast channel.
Byzantine Broadcast (BB) is a fundamental problem in distributed computing \cite{psl_BG_1982}. Consider a system of $n$ nodes, namely $p_1, ..., p_n$, of which at most $t$ nodes may be faulty. We assume that node $p_1$ is the
{\em source} for the Byzantine Broadcast (BB); the remaining $n-1$ nodes
will be referred as {\em peers}. Byzantine Broadcast must
satisfy the following three properties, assuming that source $p_1$'s input is $x$:

\begin{itemize}
\item \textbf{Termination}: every fault-free peer $p_i$ eventually decides on an output value $y_i$.

\item \textbf{Consistency}: the output values of all the fault-free peers are equal, i.e., there exists $y$ such that, for every fault-free peer $p_i$, $y_i = y$.

\item \textbf{Validity}: if the source node $p_1$ is fault-free, then the agreed value must be identical to $p_1$'s input value, i.e., $y = x$.

\end{itemize}

We assume that the input $x$ at source $p_1$ is $L$ bits long. The case when $L=1$ will be referred to as {\em Binary} BB, with the case when $L>1$ being referred to as {\em Multi-Valued} BB.

We are interested in two measures of complexity of BB algorithms.
\begin{itemize}
\item {\bf Message complexity}: Message complexity of an algorithm is defined as the maximum (i.e., worst-case) number of messages transmitted by all the nodes following the specification of the algorithm over all permissible executions.

\item {\bf Bit complexity}: Bit complexity of an algorithm is defined as the maximum number of bits transmitted by all the nodes following the specification of the algorithm over all permissible executions.
\end{itemize}

The Byzantine Broadcast problem has been considered under different
models of the communication network:

\begin{itemize}
\item {\bf Point-to-point model} (e.g., \cite{psl_BG_1982, Modular_n2_Coan_1992, gradecast_benor_2010, liang_BC_L-bit_PODC2011}): Nodes are connected via pairwise private channels,
with the network being modeled as a directed graph. Thus, if a channel $(p_i,p_j)$ exists
then node $p_i$ can transmit information to $p_j$ -- the channel is private in the sense
that no other node can overhear this transmission.

\item {\bf Reliable Broadcast model} (e.g., \cite{Koo_radio_byzantine, Vartika_radio_byzantine_2005, 2cast_journal}):
There are two models assuming the existence of reliable broadcast channels. In both models, each node has a certain set of {\em neighbors} such that a message sent by the node on the broadcast channel is received (reliably) by all its neighbors. The first model corresponds to a radio network wherein nodes within a certain distance of each node are considered its neighbors (e.g., \cite{Koo_radio_byzantine, Vartika_radio_byzantine_2005}). \cite{2cast_journal, Sen_PODC12_coloring} considers a model wherein every subset of three nodes shares a reliable broadcast channel. Importantly, in both models, the broadcast property holds for transmissions from faulty nodes as well. Thus, each transmission of a (faulty or fault-free) node is received by all its neighbors.

\end{itemize}

In this paper, we consider a new model which ``interpolates'' between the above two models. The motivation behind the new model is to understand the impact of the network assumption on performance metrics of interest. Observe that:

\begin{itemize}
\item In the reliable broadcast model, performance can be improved by exploiting the broadcast channel, since a single transmission can be received by multiple nodes. However, this model is somewhat optimistic in the sense that it does \underline{not} allow for the possibility that a faulty node may send different messages to different nodes. For instance, in a wireless setting, it
is conceivable that a faulty node may use beam-forming (or directional) antennas to send different messages to its neighbors.

\item While the point-to-point network model allows the faulty nodes to send different messages to different neighbors, it does \underline{not} provide the benefits of broadcast to the fault-free nodes.
\end{itemize}

We now introduce our {\em selective broadcast} model. To simplify the discussion and analysis, we assume a ``single-hop'' broadcast channel, as elaborated in the model below. However, the model can be extended to ``multi-hop'' networks as well.

\paragraph{Selective Broadcast Model:}
The system is assumed to be synchronous. We assume that all the $n$ nodes share a channel on which broadcasts and unicasts can both be performed. A fault-free node can broadcast its messages to all the nodes. However, a faulty node can simultaneously (in time) send different messages to different nodes to maximize the impact of its misbehavior. 

Faulty nodes cannot cause collisions on the channel. All transmissions are assumed to be reliable. While multiple nodes may transmit messages on the shared channel within the same ``round'', for the purpose of the analysis of bit and message complexity, it is adequate to assume that these transmissions are performed in parallel. (In real networks, the transmissions will be serialized by a medium access mechanism -- this detail is ignored in our work, since we
focus on bit and message complexity, not time complexity).

It is assumed that each node can correctly identify the transmitter of each message on the {\em selective broadcast} channel.

\paragraph{Failure Model:}

The Byzantine adversary has complete knowledge of the algorithm, and the source's input value. The adversary can {\em compromise} up to $t < n/3$ nodes over the entire execution of the algorithm. These compromised nodes are said to be {\em faulty}. The faulty node can arbitrarily deviate from the algorithm specification, including sending mismatching messages to other nodes.

\section{Upper Bound on Bit Complexity - Multi-Valued BB}
\label{sec:alg-L}

In this section, we present a Byzantine Broadcast algorithm for 
an $L$-bit input. The algorithm is motivated by past algorithms
that utilize ``dispute control'' or similar structures \cite{Fitzi_multi-value_BA_PODC, liang_BC_L-bit_PODC2011}.  For a suitably chosen integer $D$, the algorithm consists of $L/D$ ``generations'', with Byzantine Broadcast of $D$ bits of the input being achieved in each generation $g$ ($1 \leq g \leq L/D$). For simplicity, we assume that $L/D$ is
an integer. Input $x$ at source node $p_1$ is viewed as the tuple
\[ x(1), x(2), \cdots, x(L/D), \]
where each $x(g)$ consists of $D$ bits. In each generation $g$, each fault-free node $p_i$ obtains output $y_i(g)$, and its $L$-bit output for all the generation together is obtained as the tuple
\[ y_i(1), y_i(2), \cdots, y_i(L/D)\]


Our {\em Byzantine Broadcast} algorithm has three
phases: {\em Detectable Broadcast}, {\em Detection Dissemination}
and {\em Dispute Control}.

\bigskip

\vspace*{8pt}\hrule
\noindent {\bf Byzantine Broadcast using Dispute Control}
\vspace*{4pt}\hrule

~

Perform the following steps in the $g$-th generation, $g=1,\cdots,L/D$. 
\begin{itemize}
\item {\bf Detectable Broadcast of $x(g)$}:
By the end of the {\em Detectable Broadcast} phase each fault-free node $p_i$
receives a value, denoted as $z_i$, such that one of the following two conditions is true:
\begin{list}{}{}
\item[(i)] at least one fault-free node detects misbehavior by some faulty node(s) in the network, without necessarily identifying the faulty node(s), or

\item[(ii)] no fault-free node detects any misbehavior, and for all fault-free
	peers $p_i$ and $p_j$, we have 
		$z_i=z_j$, and
	additionally, if $p_1$ is fault-free then $z_i=z_j=x(g)$.
\end{list}

\item {\bf Detection Dissemination:}
Each node performs a Byzantine broadcast (BB) of a single bit, indicating
whether it detected any misbehavior during {\em Detectable Broadcast} phase or not.
The 1-bit BB is performed using any previously proposed BB algorithm (e.g., \cite{psl_BG_1982, Modular_n2_Coan_1992, gradecast_benor_2010}) for the point-to-point model (note that algorithms designed for the point-to-point model
can be used under our {\em selective broadcast} model as well).

If no node announces that it detected misbehavior, then each fault-free node $p_i$
sets output $y_i(g)$ equal to $z_i$ received in the {\em Detectable Broadcast} phase,
and the {\em Dispute Control} phase is \underline{not} performed in the $g$-th generation.

\item {\bf Dispute Control:} If any node indicates,
 in the {\em Detection Dissemination} phase, that misbehavior has been detected,
 then additional steps are taken to learn
new information regarding the potential identity of the faulty nodes. In particular,
a new pair of nodes is found ``\underline{in dispute}" such that at least one node in this
pair is guaranteed to be faulty. Two fault-free nodes are never
found {\em in dispute} with each other.
The {\em Dispute Control} phase is performed using a
 BB algorithm previously proposed for the point-to-point model
(similar to the {\em Detection Dissemination} phase), and as a by-product
of {\em Dispute Control}, Byzantine Broadcast of $x(g)$ is also achieved.
For brevity, we omit the details of the {\em Dispute Control} phase; similar
dispute control mechanisms have been included in prior work as well \cite{Fitzi_multi-value_BA_PODC,liang_BC_L-bit_PODC2011}.

For future reference, note that, any node that is found {\em in dispute}
with more than $t$ other nodes must necessarily be faulty itself. Such a
node is then essentially excluded from the future generations of the algorithm.
If the source node $p_1$ is thus identified as faulty, then the algorithm
terminates with the nodes agreeing on a default value for all future
generations.
\end{itemize}

\hrule

~

As shown in prior work using {\em Dispute Control}, the bit complexity of such algorithms is dominated by the first phase -- {\em Detectable Broadcast} in our case -- when the input size $L$ is sufficiently large. For brevity, we will omit the proof of this claim; the reader is referred to prior work (e.g., \cite{liang_BC_L-bit_PODC2011}) for examples of similar algorithms. There are two key reasons for this outcome:
\begin{itemize}
\item Although the 1-bit BB in the {\em Detection Dissemination} phase is performed using an expensive (existing) algorithm, the cost of these disseminations is amortized over a large number of bits, specifically $D$, in the input for each generation.

\item Similarly, the {\em Dispute Control} phase is also expensive. However, it turns out that this phase is performed only a finite number of times (specifically, at most $t(t+1)$ times). Thus, its amortized cost over a large number of generations (i.e., large $L/D$) is small.
\end{itemize}
In the rest of this section, we present and analyze a {\em Detectable Broadcast}
algorithm for $D$ bits. 

\subsection{$D$-Bit Detectable Broadcast}

The {\em Detectable Broadcast} phase makes use of an error detection code. Therefore,
we first describe the code and its parameters.

\paragraph{Error Detection Code:}

With a suitable choice of parameter $c$, we will use a $(n,n-2t)$ Reed-Solomon code over Galois Field GF($2^c$). In particular, $c$ is chosen large enough such that $n \leq 2^c - 1$. The $D$-bit value to be agreed on in each generation is viewed as consisting of $n-2t$ {\em data} symbols from $GF(2^c)$. Thus,
each of these symbols can be represented with $c$ bits, and therefore, $D=c(n-2t)$. Given the $(n-2t)$ data symbols corresponding to a certain $D$-bit
value, $n$  ``\underline{coded}'' symbols in the corresponding codeword are obtained as linear independent combinations of the $(n-2t)$ data symbols over $GF(2^c)$. The code specification is part of the specification of the {\em Detectable Broadcast} algorithm.  The $(n,n-2t)$ Reed-Solomon code has the following useful property: Any $n-2t$ (coded) symbols in a codeword can be used to compute the corresponding $n-2t$ data symbols, and therefore, the corresponding $D$-bit value. We summarize the relationships between the code parameters:
\begin{itemize}
\item $n$ coded symbols in each codeword, corresponding to $(n-2t)$ data symbols,
\item $n \leq 2^c-1$, and $D=c(n-2t)$
\end{itemize}
This implies that $n \leq 2^{D/(n-2t)}-1$, and $D\geq (n-2t)\log_2(n+1)$.
Thus, we need $D = \Omega(n \log n)$.


\paragraph{Detectable Broadcast:}

Algorithm 1 below specifies execution of {\em Detectable Broadcast} in the $g$-th generation of the {\em Byzantine Broadcast} algorithm described above. Recall that $x(g)$ is the input for source $p_1$ in the $g$-th generation.

\vspace*{8pt}\hrule
\noindent {\bf Algorithm 1: Detectable Broadcast for the $g$-th Generation}
\vspace*{4pt}\hrule

\begin{enumerate}
\item
\label{step_1} Source $p_1$ transmits
 $x(g)$ on the {\em selective broadcast} channel.
$x(g)$ is viewed as a vector of $n-2t$ {\em data} symbols, each
consisting of $c=D/(n-2t)$ bits each.

\item Peer $p_i$ ($2\leq i\leq n$) performs the following steps:


\begin{enumerate}
\item \label{step_2_transmit} If $p_i$ is in dispute with the source, stay silent.

Else, encode the $n-2t$ data symbols received in step \ref{step_1} into a codeword consisting of $n$ {\em coded} symbols using the aforementioned coding scheme. Denote the resulting codeword obtained at peer $i$ as $s_{i} = s_i[1], s_i[2],\cdots, s_i[n]$. Transmit $s_i[i]$ on the {\em selective broadcast} channel.
\item \label{step_2}
Node $p_i$ ignores symbols received from peers that it has been
previously
{\em in dispute} with. For each peer $p_j$ that node $p_i$
is {\em not} in dispute with:
let the symbol received from peer $p_j$ be denoted as $r_i[j]$. Also,
if $p_i$ is {\em not in dispute} with $p_1$, then let $r_i[1]=s_i[1]$. 
For all nodes $p_k$ such that $p_i$ is in dispute with $p_k$ ($1\leq k\leq n$),
define $r_i[k]=\perp$ ($\perp$ denotes a distinguished {\em null} symbol).
\end{enumerate}

At the end of step \ref{step_2}, every peer $p_i$ has one
non-null symbol $r_i[j]$
corresponding to each node $p_j$ ($1\leq j\leq n$) that
has not yet been dispute with either $p_i$ or $p_1$.
Since $p_1$ and $p_i$ can each be in dispute with at most $t$ nodes
(otherwise they would have already been identified as faulty), 
node $p_i$ has at least $n-2t$ non-null symbols in the $r_i$ vector.

\item \label{step_3} Peer $p_i$ ($2 \leq i \leq n$) performs the following steps:


\begin{enumerate}
\item Find the solution for each subset of $n-2t$ non-null coded symbols in the
$r_i$ vector received above in step \ref{step_2} -- the solution consists of
$n-2t$ data symbols that correspond to the $n-2t$ coded symbols. 
\item If the solutions to all these subsets of size $n-2t$ is not unique, then $p_i$ has detected faulty behavior by some peer. In this case, $z_i$ is set equal to some default value.

Else, $z_i$ is set equal to the unique solution corresponding to any of the $n-2t$ non-null symbols received in step \ref{step_2}. The solution consists
of $n-2t$ symbols, which correspond to $c(n-2t)=D$ bits.

\end{enumerate}

\comment{
Else, $p$ performs the following steps:

\begin{enumerate}
\item For $2\leq j\leq n$, if $p_i$ is not in dispute with $p_j$ and $p_j$ is not in dispute with the source, then compare $r_i[j]$ and $s_i[j]$. Any other messages are simply ignored.
\item
If all comparisons result in a match,
then $z_i$ is set equal to the $D$ bits corresponding to the $n-2t$ symbols
received in step \ref{step_1}. 

Else, peer $p_i$ has detected faulty behavior by some node. In this
case, $z_i$ is set equal to some default value.
\end{enumerate}
}
\end{enumerate}
\hrule

~

The bit complexity of {\em Detectable Broadcast} algorithm can be further reduced by using a more efficient coding scheme. For lack of space, we omit the discussion here.

\paragraph{Correctness of Detectable Broadcast Algorithm}

\begin{lemma}
\label{lemma:detectable}
By the end of Detectable Broadcast in $g$-th generation,
the following conditions hold:
\begin{itemize}
\item either at least one fault-free node
detects misbehavior by some faulty nodes,
\item or for all fault-free peers $p_i$ and $p_j$,
$z_i=z_j$, and additionally, if $p_1$ is fault-free then
$z_i=z_j=x(g)$. 
\end{itemize}
\end{lemma}

\begin{proof}
In the {\em Detectable Broadcast} algorithm, the following misbehaviors
are possible:
Source $p_1$ may misbehave in step \ref{step_1}
by transmitting different $D$-bit values (represented as $n-2t$ symbols)
to at least two different fault-free peers that are both not in dispute with the source. A peer node may misbehave
by transmitting incorrect symbols to some fault-free peer(s) in step \ref{step_2}.

Consider two cases in the $g$-th generation:
\begin{itemize}
\item {\em Source does not misbehave: }: Since there are
at least $n-t$ fault-free nodes, and fault-free nodes are never in
dispute with each other, it is easy to see that, for each pair
of fault-free peers $p_i$ and $p_j$, the vectors $r_i$ and $r_j$
will include at least $n-t$ identical and correct coded symbols
corresponding to the data symbols sent by $p_1$ in step \ref{step_1}.
Then it should be easy to see that for each pair of fault-free peers
$p_i$ and $p_j$, either (i) at least one of them will
detect misbehavior by some node, or (ii) $z_i=z_j=x(g)$.

\comment{
  In this case, 
in step \ref{step_1}, all fault-free peers receive identical $n-2t$ symbols
representing the $D$-bit value $x(g)$.
Consider two sub-cases:
\begin{itemize}
\item No peer misbehaves:
Then it is easy to see that
each fault-free peer $p_i$ will not detect any mismatch in step \ref{step_3},
and therefore set $z_i=x(g)$. 
\item A peer misbehaves:
If some peer $p_k$ misbehaves by sending an incorrect symbol to fault-free
peer $p_i$, then $p_i$ will detect that $s_i[k] \neq r_i[k]$, and thus detect
misbehavior. Note that node $p_i$ cannot conclude that $p_k$ is the
faulty node, since it is possible that the source sent incorrect message
to $p_k$ in the previous step. 
\end{itemize}
}

\item {\em Source node misbehaves:} In this case, the source node is in dispute with at most $t$ nodes (otherwise, it would have been identified as faulty already). Thus, the source node is not in dispute with at least $(n-1)-t$ peers.
Of these, $n-1-t$ peers, at least $(n-1-t)-(t-1)=n-2t$ peers are fault-free.
Since these fault-free peers cannot in dispute with each other, they will
collectively transmit at least identical $n-2t$ symbols to all the nodes.
Therefore, all the fault-free peers will share at least $n-2t$
symbols in common in their $r$ vectors.
Then it should be easy to see that for each pair of fault-free peers
$p_i$ and $p_j$, either (i) at least one of them will
detect misbehavior by some node, or (ii) $z_i=z_j$.

\end{itemize}
\end{proof}

Lemma \ref{lemma:detectable} together with the correctness of {\em Detection Dissemination} and {\em Dispute Control} (similar to the proofs in \cite{liang_BC_L-bit_PODC2011}) proves the correctness of the algorithm presented in this section.


\subsection{Bit Complexity}

Source node $p_1$ transmits $D$ bits in step \ref{step_1}.
In step \ref{step_2_transmit}, at most $n-1$ peers
each transmit a coded symbol
consisting of $D/(n-2t)$ bits, for a total cost of
$(n-1)D/(n-2t)$ bits in step \ref{step_2_transmit}.
Thus, the worst-case cost of a single instance of {\em Detectable Broadcast}
(in bits)
is
\[
D + \frac{(n-1)D}{n-2t}
\]
Thus, the total cost of {\em Detectable Broadcast} over all the $L/D$
generations required to perform Byzantine Broadcast of the $L$-bit input
at node $p_1$ is given by
\[
L + \frac{(n-1)L}{n-2t} ~ = ~ L \, \frac{2n-2t-1}{n-2t}
\]

As noted previously, with the {\em dispute control} framework, when $L$
is large, the bit complexity of Byzantine Broadcast of $L$-bits
is dominated by the bit complexity of {\em Detectable Broadcast}. In
particular, it can be shown that if we use the Modular Algorithm \cite{Modular_n2_Coan_1992} to perform 1-bit BB in {\em Detection Dissemination} and {\em Dispute Control} phase, then
the communication cost (in bits) of the proposed {\em Byzantine Broadcast}
algorithm is
\[
L \, \frac{2n-2t-1}{n-2t} ~+~ O(n^4) L^{0.5}
\]

When $L$ is large, the first term dominates the above cost, and the
bit complexity becomes $O(L)$. By assumption, $n\geq 3t+1$, and therefore,
\[
2 < \frac{2n-2t-1}{n-2t} < 4
\]

\section{Upper Bound on Message Complexity}
\label{sec:alg-tlogt}

In this section, we briefly describe a Byzantine Broadcast algorithm, named Algorithm 2, with message complexity $O(t \log t)$ under the {\em selective broadcast} model. This algorithm is
derived from the {\em Modular Algorithm} proposed by Coan and Welch \cite{Modular_n2_Coan_1992}, which achieves optimal message complexity of $O(n^2)$ in the point-to-point model. Algorithm 2 has the same structure as the {\em Modular Algorithm} except for the following modifications:

\begin{itemize}
\item In the {\em Modular Algorithm}, a node may often send identical messages to multiple neighbors. In our algorithm, such transmissions are replaced by {\bf one} broadcast message over the {\em selective broadcast} channel.


\item The {\em Modular Algorithm} is performed by all the $n$ nodes. In our case, we use a modified version of that algorithm to achieve Byzantine Broadcast among $3t+1$ nodes. After that, any $2t+1$ of these $3t+1$ nodes transmit the agreed value on the {\em selective broadcast} channel; the remaining $n-3t-1$ nodes agree on a majority vote of these transmissions.
\item The number of levels of recursion in our algorithm is $\log n$, different from
that in the original {\em Modular Algorithm} \cite{Modular_n2_Coan_1992}.
\end{itemize}

Now, we briefly present the framework from \cite{Modular_n2_Coan_1992} and its
variation for the {\em selective broadcast} model.

\subsection{Modular Framework \cite{Modular_n2_Coan_1992}}

The framework proposed by Coan and Welch \cite{Modular_n2_Coan_1992} consists of recursive applications of two transformations, BC2BCB and BCB2BC: the first transformation, namely BC2BCB, uses a Byzantine Consensus (BC) algorithm to solve a Byzantine Committee Broadcast (BCB) problem, and the latter transformation, namely BCB2BC, uses a BCB algorithm to solve the BC problem.

The base case of the Modular Framework is a previously proposed BC algorithm (e.g., \cite{psl_BG_1982, gradecast_benor_2010}), say $A_0$. Then, the recursive definition of the Modular Framework is as follows: Given algorithm $A_{i-1}$ ($i \geq 1$), $A_i$ is defined as BCB2BC(BC2BCB($A_{i-1}$)). For brevity, we omit the details of the framework; the reader is referred to the prior work \cite{Modular_n2_Coan_1992}. The main reason that the Modular Framework achieves low message complexity is as that
the number of message transmissions induced by the expensive base algorithm is small. This is achieved by recursively dividing the nodes into many ``committees'' of small size such that the base algorithm is only executed within each of these small committees.

\paragraph{Algorithm 2:}

Modular Algorithm \cite{Modular_n2_Coan_1992} is designed for the Byzantine Consensus (BC) problems, wherein each node has an input. To perform Byzantine Broadcast (BB) using a modified Modular algorithm, the source node first broadcasts its input value on
the {\em selective broadcast} channel, and then a modified Modular algorithm is used to reach consensus on the value received by all the peers from the source . Now, we describe two modifications of Modular Algorithm to achieve BC efficiently in the {\em selective broadcast} model.

In the Modular Algorithm, when fault-free nodes transmit, they always transmit the same message to all the intended receivers. Thus, replacing these transmissions by a {\bf single} transmission in {\em selective broadcast} channel reduces the message complexity.

The second modification further exploits the reliable broadcast channel. Unlike Modular Algorithm, Algorithm 2 only requires $n' = 3t+1$ {\em active} nodes participating in the algorithm, i.e., executing Algorithm 2, due to the existence of {\em selective broadcast} channel. The
remaining $n-3t-1$ nodes (which we call {\em passive} nodes) do not transmit messages at all, but listen to the messages announcing the agreed value, transmitted by any $2t+1$ {\em active} nodes. The passive nodes then use majority voting on these $2t+1$ values to decide on their output.

\subsection{Message Complexity}

Suppose that $M_*(n')$ is the message complexity of the base algorithm executed by $n'$ nodes.
For any fixed integer $B$ such that $t+1 \geq B \geq 2$, by using analysis similar to that in \cite{Modular_n2_Coan_1992}, we can show the upper bound below on message complexity of $A_i$, denoted as $M_i$, for all $\log_{B} t \geq i \geq 0$. In the inequality below, $\alpha$ is a certain constant that depends
on the transformation used in the Modular Algorithm.

\begin{equation}
\label{eq:Ai}
M_i(3t+1) \leq B^i M_*(3t/B^i + 1) + \alpha B t i~,
\end{equation}

Replacing $i = \log_{B} t$ and $M_*(n') = O((n')^3)$ in equation \ref{eq:Ai} yields $M_{\log_{B} t}(3t+1) \leq O(t \log t)$. \footnote{Here, we use the Gradecast-based algorithm \cite{gradecast_benor_2010} as the base algorithm.}

The message complexity of Algorithm 2 is $M_{\log_{B} t}(3t+1) + 2t+1$, and is thus bounded by $O(t \log t)$.

\comment{
\begin{equation*}
M(\text{Algorithm 2}) = M_{\log_{B} t}(3t+1) \leq O(t \log t)
\end{equation*}
}

\section{Lower Bounds}
\label{sec:LB}

In this section, we state some simple lower bounds on bit and message
complexity. In deriving these bounds, we assume that the nodes only communicate explicitly through messages. That is, no implicit communication mechanism is used to convey information, such as the time between two message transmissions.

\paragraph{Lower Bound on Message Complexity:}

\begin{theorem}
The lower bound on message complexity of Byzantine Broadcast under the {\em selective broadcast} model is $\Omega(t)$.
\end{theorem}

\begin{proof}
The proof is by contradiction.
Assume that there exists a correct algorithm A under the {\em selective broadcast} model that has message complexity of $o(t)$.
The transmission of a message by a certain node on the {\em selective broadcast} channel
can be simulated by sending a copy of the message to each of the remaining $n-1$ nodes
in a fully connected point-to-point network. Thus, algorithm A can solve the Byzantine
Broadcast problem using $o(nt)$ messages in a point-to-point model. This contradicts
the lower bound in \cite{dolev_msg_complexity}.

Alternatively, the theorem can be proved by arguing that more than
$t$ messages must be sent in the worst-case: if only $t$ (or fewer) messages are sent,
it is possible that the transmitters (at most $t$) of all these messages are faulty nodes, making it impossible to
guarantee agreement among the fault-free nodes. 
\end{proof}

The gap between the above lower bound of $\Omega(t)$ and the upper bound
of $O(t \log t)$ in Section \ref{sec:alg-tlogt} is small, but the problem of closing this gap remains open.

\paragraph{Lower Bound on Bit Complexity:}

It should be obvious that the total number of bits transmitted in any BB algorithm is at least $L$ (in the worst case), since there are $2^L$ possible input values. Thus, $\Omega(L)$ is a trivial lower bound on bit complexity, and our algorithm presented in Section \ref{sec:alg-L} matches this bound when $L$ is large enough.

\paragraph{Lower Bound on Bit Complexity of ``Static'' Algorithms
for Detectable Broadcast:}


An algorithm A is characterized by a set of {\em schedules} where each schedule consists of a sequence of transmission slots. In each slot $i$, a single node is selected as the {\em transmitter}, denoted as $T_i$. The transmitter $T_i$ transmits a message via the {\em selective broadcast} channel to all the other nodes. An algorithm is said to be {\bf static} if it has a fixed schedule, such that the transmitters in all the slots are pre-determined and are independent of the source's input and the behavior of the faulty nodes. 


We state the following lower bound on bit complexity of
{\em static} algorithms for {\em Detectable Broadcast} (DB). 
For lack of space, we only present the case when no two nodes are in dispute, and the proof is omitted. The proof argues that
in a static schedule, the source $p_1$ must transmit at least $L$ bits,
and then argue that the remaining nodes must transmit at least
$(n-1)\frac{L}{n-f}$ bits to satisfy the conditions of {\em Detectable
Broadcast}.

\begin{claim}

The lower bound on bit complexity of static DB algorithm is \footnote{Note that this lower bound is also a lower bound of static BB algorithms, since any BB algorithm also solves DB problems. Note that algorithms in \cite{psl_BG_1982, Modular_n2_Coan_1992, gradecast_benor_2010} are static BB algorithms, but the algorithm presented in Section \ref{sec:alg-L} is not static.}

\begin{equation*}
\label{eq:equality_check}
L + (n-1) \frac{L}{n-f}
\end{equation*}
\end{claim}

\section{Summary}

This paper introduces a new communication model that ``interpolates'' between two old models in the literature. In particular, we explore the impact of allowing nodes to {\em select} between broadcast and unicast on bit and message complexity. In the new model, we present a Multi-Valued BB algorithm that is order-optimal in bit complexity, and another BB algorithm that is efficient in message complexity. At last, we briefly discuss about lower bounds on bit and message complexity in the new model.

\bibliography{paperlist}

\end{document}